\documentclass[aps,prx,floatfix,nofootinbib,superscriptaddress,longbibliography]{revtex4-2}
\usepackage{appendix}
\usepackage{natbib}
\usepackage{graphicx}
\usepackage{amsmath}
\usepackage{amsthm}
\usepackage{amsfonts}
\usepackage{subfigure}
\usepackage{dsfont}
\usepackage{eufrak}
\usepackage{tikz}
\usepackage{braket}
\usepackage{listings}
\usepackage[margin=0.8in]{geometry}
\usepackage{quantikz}
\usepackage[normalem]{ulem}
\usepackage{pifont}
\newcommand{\cmark}{\ding{51}}%
\newcommand{\xmark}{\ding{55}}%

\newtheorem{lemma}{Lemma}
\newtheorem{theorem}{Theorem}

\begin{document}

\title{Absence of censoring inequalities in random quantum circuits}

\def\urbana{
Institute for Condensed Matter Theory and IQUIST and NCSA Center for Artificial Intelligence Innovation and Department of Physics, University of Illinois at Urbana-Champaign, IL 61801, USA
}

\author{Daniel Belkin}
\thanks{Co-first authors.}
\affiliation{\urbana}
\author{James Allen}
\thanks{Co-first authors.}
\affiliation{\urbana}
\author{Bryan K. Clark}
\email{bkclark@illinois.edu}
\affiliation{\urbana}
 
\begin{abstract}
    Ref.~\onlinecite{Harrow2023} asked whether deleting gates from a random quantum circuit architecture can ever make the architecture a better approximate $t$-design. We show that it can. In particular, we construct a family of architectures such that the approximate $2$-design depth decreases when certain gates are deleted. We also give some intuition for this construction and discuss the relevance of this result to the approximate $t$-design depth of the 1D brickwork.  Deleting gates always decreases scrambledness in the short run, but can sometimes cause it to increase in the long run. Finally, we give analogous results for spectral gaps and when deleting edges of interaction graphs. 
\end{abstract}

\maketitle

\section{Introduction}
Imagine you wish to mix two colors of paint together in a bucket. If you stir only one side of the bucket, the colors won't mix perfectly, but one expects that the end state will at least be closer to mixed than where you started. In general, one expects a local scrambling operation to increase the overall scrambledness of a physical system. Formalizations of this intuition are sometimes known as ``censoring inequalities". One needs to choose a set of local operations (e.g. Glauber dynamics) and find a way to quantify scrambledness (e.g. total variation distance from the equilibrium distribution). Such inequalities have been established for many classical Markov chains of interest, usually via monotonicity arguments\cite{Peres2011, Lacoin2016, Gantert2022, Caputo2019, Fill2013, Schmid2019, Nestoridi2023, Blanca2018}, but are known not to hold in all cases \cite{Holroyd2011}. 

Ref.~\onlinecite{Harrow2023} asked whether deleting gates from a random quantum circuit architecture can ever make the corresponding ensemble a better approximate $t$-design. This question is equivalent to the existence of a censoring inequality for random quantum circuits, with the local operations being Haar-random local unitary gates and ``scrambledness'' quantified by distance from the global Haar measure. It seems intuitively plausible that adding more local scrambling must increase global scrambledness in this setting, but a proof has been elusive. 

One particular censored circuit of interest is the architecture introduced in Ref.~\cite{Schuster2024}. This circuit can be obtained from the 1D brickwork by deleting $O\left(\frac{N}{\log N}\right)$ gates, and it is known to form an approximate $t$-design in depth at most $O(\log N)$. For the brickwork, meanwhile, the best known upper bound is $O(N)$. A censoring inequality for multiplicative error would thus imply a large improvement to the approximate $t$-design depth of the 1D brickwork. Another case of interest is Conjecture 1 of Ref.~\onlinecite{Belkin2023}, which proposes that the singular value spectral gaps of arbitrary architectures may be bounded by that of the brickwork. A censoring inequality for singular value gaps would imply that spectral gaps of arbitrary architectures are at least bounded by those of treelike architectures. 
 
In this work we give a counterexample to show that no such censoring inequality can exist. We construct a circuit ensemble which is not an $\epsilon$-approximate $t$-design, but becomes one when certain gates are deleted. Our proof computes the spectral gaps of the $t=2$ moment operator and relates these spectral gaps to approximate $t$-design conditions. We also give some intuition for our construction via a classical analogy (Section \ref{section:intuition}) and present numerical results (Figure \ref{fig:architecture} and Appendix \ref{app:numerics}).

\section{Main results}
Consider a circuit $C$ on $N$ qubits which consists of five Haar-random unitary gates. The first gate $U_1$ acts on every qubit except the first. The second gate $U_2$ acts on the first and second qubits. The third gate $U_3$ acts on every qubit except the second. $U_4$ and $U_5$ act on the same sites as $U_2$ and $U_1$, respectively, so that the circuit ensemble is symmetric under time reversal. Let $C^{d}$ refer to the circuit architecture obtained by repeating this architecture $d$ times. We are interested in comparing $C$ against an analogous architecture $C'$ which omits gate $U_2$ and $U_4$. These two circuit architectures are illustrated in Figure \ref{fig:architecture}a.

\begin{theorem}
    \label{thm:depth}
    For sufficiently large $N$, there exist values of $\epsilon$ and $d$ such that $C'^d$ is an $\epsilon$-approximate $2$-design and ${C}^{d}$ is not. 
\end{theorem}

In other words, deleting gates can cause a circuit to form an approximate $t$-design faster. This result is true for both multiplicative-error and additive-error $t$-designs. An analogous result holds when scrambledness is measured by the eigenvalue spectral gap of the moment operator, which controls the asymptotic rate of scrambling in a periodic architecture \cite{Allen2024}, or by the singular value spectral gap, which is related to the scrambling rate in arbitrary aperiodic architectures \cite{Belkin2023}.  

\begin{figure}[h]
    \centering
    \begin{tikzpicture}
        \begin{scope}
            \node[anchor=north west,inner sep=0] (image_a) at (0,0.45)
            {\includegraphics[width=0.34\columnwidth]{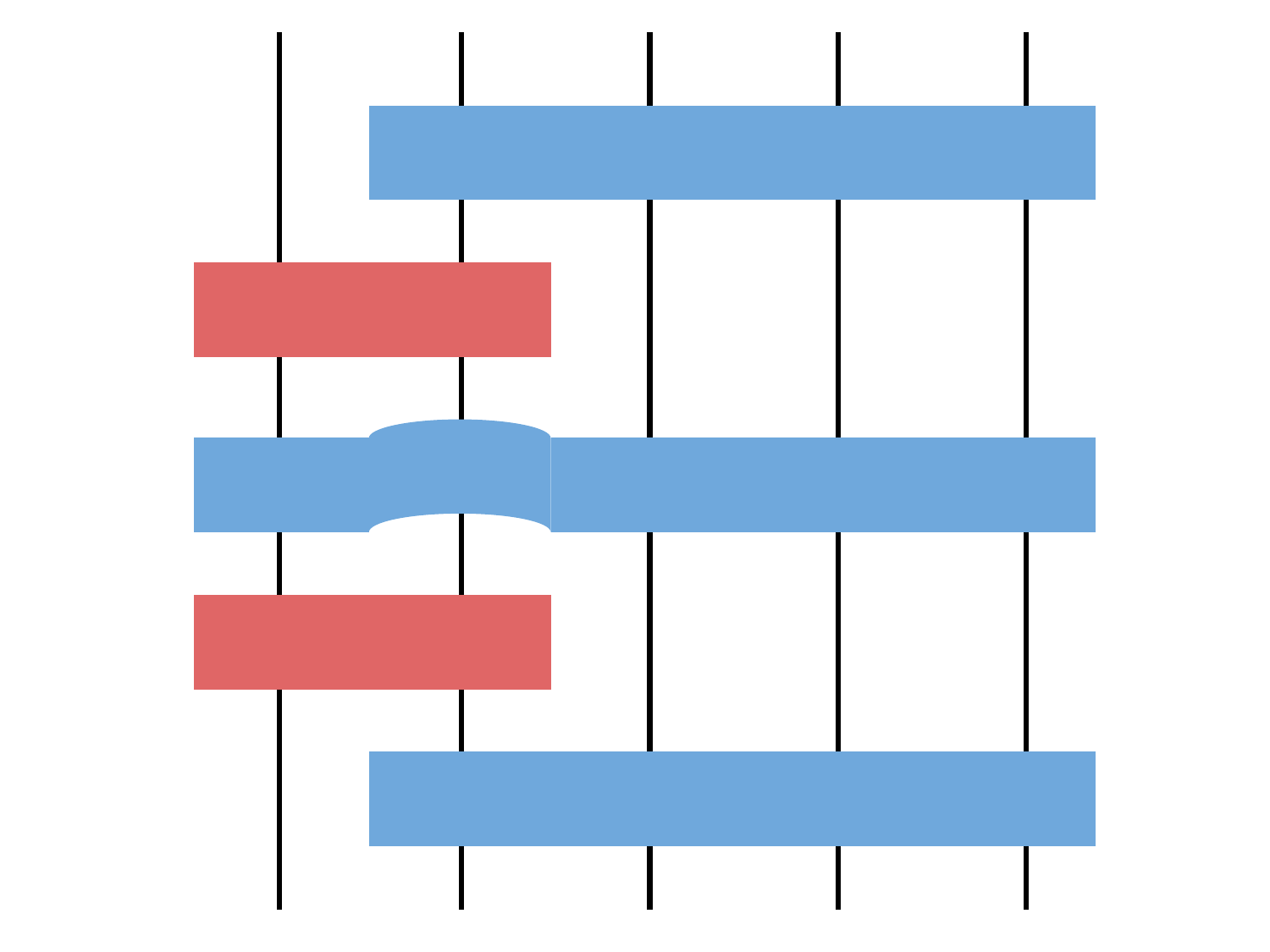}};
            \node [anchor=north west] (note) at (0.4,0) {\small{\textbf{a)}}};
        \end{scope}
        \begin{scope}[xshift=0.33\columnwidth]
            \node[anchor=north west,inner sep=0] (image_b) at (0,0)
            {\includegraphics[width=0.64\columnwidth]{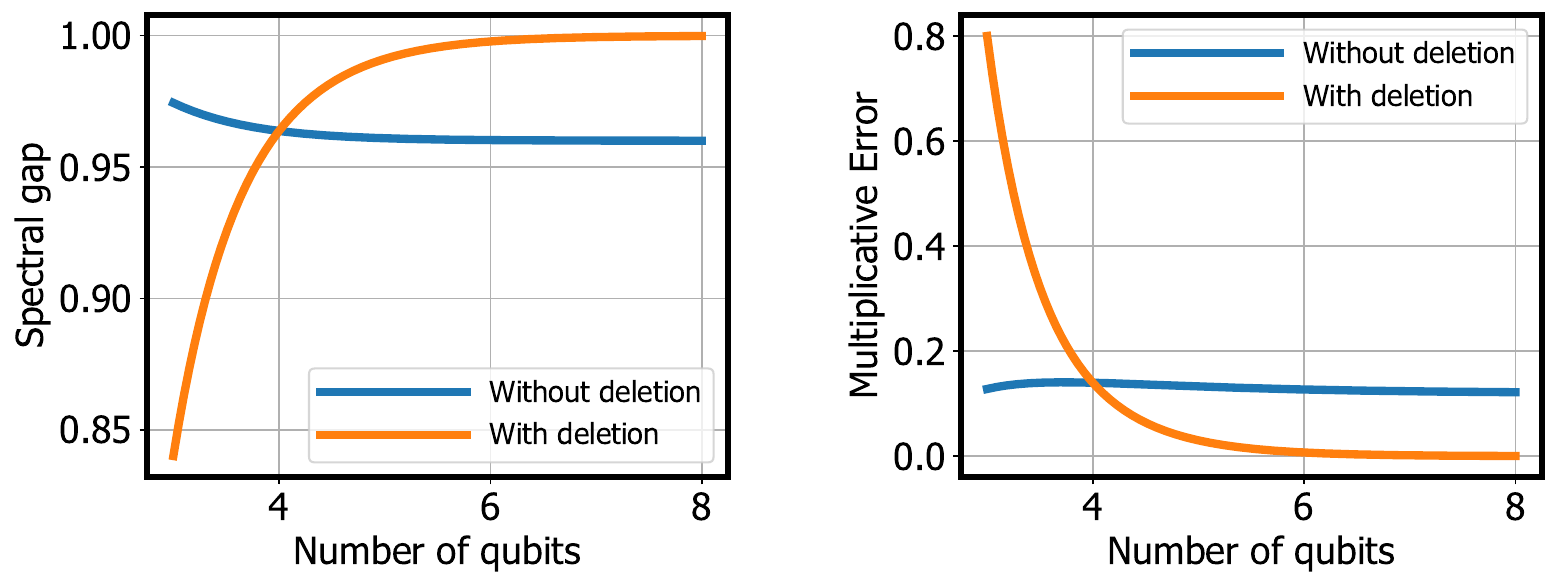}};
            \node [anchor=north west] (note) at (-0.2,0) {\small{\textbf{b)}}};
            \node [anchor=north west] (note) at (5.8,0) {\small{\textbf{c)}}};
        \end{scope}
    \end{tikzpicture}
    \vspace*{-0.4cm}
    \caption{(a) A single period of the architecture on $N = 5$ sites. When the red gates are deleted, the ensemble becomes more scrambled. (b) Spectral gap with and without red gates. (c) Multiplicative error for the second moment operator of a single period of the architecture, with and without red gates.}
    \label{fig:architecture}
\end{figure}
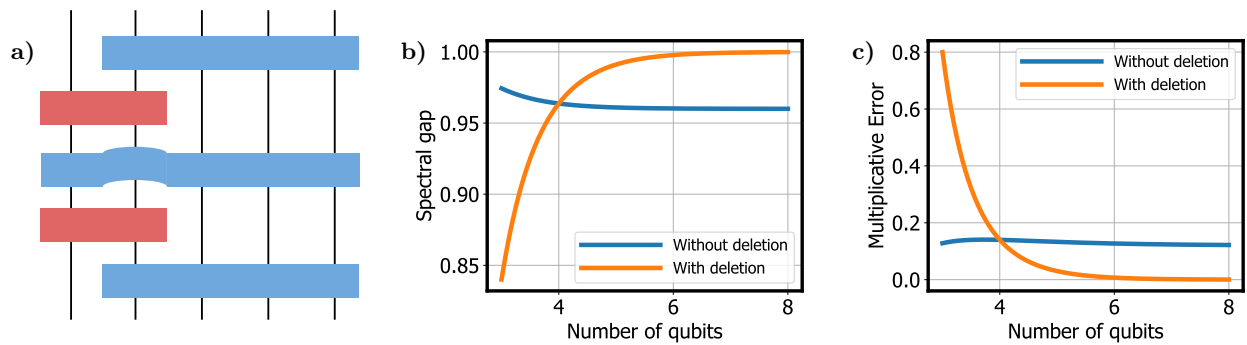

These circuits involve many-site gates and multiple deletions. However, there also exist circuits consisting only of two-site gates with the same properties under deletion of a single gate. In addition, we show numerically that censoring inequalities do not hold for non-deterministic architectures with gate locations sampled independently from the edges of some fixed graph. Our full results on the existence of censoring inequalites are summarized in Table \ref{table:status}.

We also note that censoring inequalities do hold for deletion of gates from the first or last layer of an architecture. In other words, deleting gates always moves the circuit further from equilibrium in the short run, even though it sometimes causes the circuit to equilibrate faster in the long run. This creates a crossover reminiscent of the Mpemba effect (see Fig.~\ref{fig:mpemba} in Appendix~\ref{app:numerics} for an example).

\begin{table}[h]
\label{table:status}
\begin{tabular}{|l|c|c|c|}
\hline
                            & Arbitrary Gates & Boundary Gates & Graph edges \\ \hline
Eigenvalue Spectral Gap     & \xmark             & \xmark             & \xmark          \\ \hline
Singular Value Spectral Gap & \xmark            & \cmark             & \xmark          \\ \hline
Additive Error              & \xmark            & \cmark             & \xmark          \\ \hline
Multiplicative Error        & \xmark            & \cmark             & \xmark          \\ \hline
\end{tabular}
\caption{The existence of a censoring inequality for each of four scrambling rate measures and three types of censorship. ``Boundary gates'' are gates from the first or last layer of an architecture, while ``graph edges'' refers to non-deterministic architectures in which gate positions are sampled uniformly from a graph. \cmark means we show a censoring inequality holds, while \xmark\, indicates that we construct a counterexample. Proofs are found in Appendix \ref{app:table_proofs}.}
\end{table}

\section{Intuition}
\label{section:intuition}
We begin with intuition for Theorem \ref{thm:depth}. Suppose site 1 is initially entangled with an auxiliary system $A$. Then $U_{2}$ will cause site $2$ to be entangled with $A$. Since $U_3$ does not act on site $2$, some entanglement with $A$ will remain localized on site $2$. If we remove $U_{2}$, on the other hand, the entanglement with $A$ is instead spread over every site except $2$ by gate $U_3$. Similarly, $U_4$ moves the entanglement back to site $1$ so that it is protected from $U_5$. 

In addition, a classical analogy can be made in terms of mixing colors. Suppose there is initially some red pigment localized on site 1 (Figure \ref{fig:color-mixing}) and each gate uniformly mixes the sites to which it is applied. Gates $U_2$ and $U_4$ allow a portion of the red color to move from site $1$ to site $2$ and back again, respectively, so that after the circuit $\sim 25\%$ of the original red remains localized on site 1. When we omit the even-numbered gates, on the other hand, $U_3$ dilutes the original red over every site but one. This corresponds to quantum information in the initial state being scrambled over many sites. We see that $U_2$ and $U_4$ serve to hide some of the initial information from the scrambling effect of $U_3$.

\begin{figure}[h]
    \centering
    
    \begin{tikzpicture}
        \begin{scope}
            \node[anchor=north west,inner sep=0] (image_a) at (0,0)
            {\includegraphics[width=0.6\columnwidth]{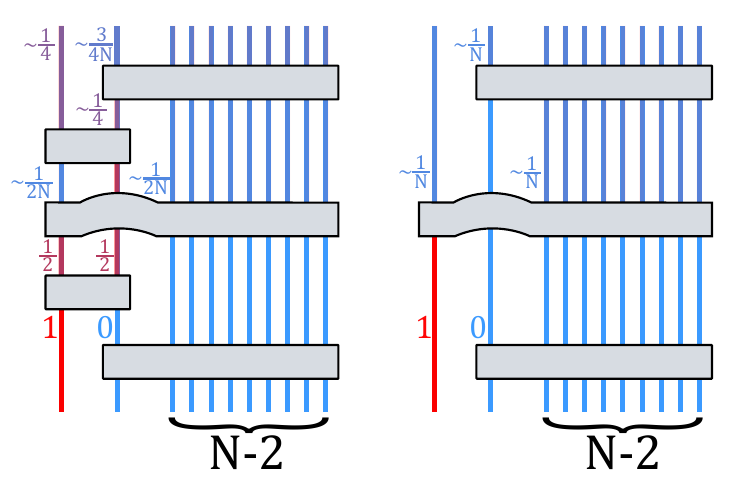}};
        \end{scope}
    \end{tikzpicture}
    \vspace*{-0.4cm}
    \caption{Classical analogy for the preservation of quantum information in the circuits $C$ \textit{(left)} and $C'$ \textit{(right)}. The scrambling of quantum information through gates is analogous to the mixing of a pigment over the whole system. We start with all the red pigment on the 1st site of the system, and track how much of that pigment gets spread to the other sites by the action of each gate. Each gate is a perfect mixer on its support, i.e. it replaces the amount of pigment on each site in its support with the average across the support. We assume $N \gg 1$ and give all pigment levels to leading order in $N$. Despite the fact that $C$ has more gates, those extra gates allow it to temporarily store the pigment on the 2nd site. As a result, the level of mixing is not as strong as it is in $C'$, where the mixing is so strong that the first site only differs from the rest of the system by an $O(\frac{1}{N^2})$ amount of pigment.}\label{fig:color-mixing}
\end{figure}

\section{Proof of Theorem \ref{thm:depth}}
Let $G_i$ be the vectorization of the $2$\textsuperscript{nd} moment operator for $U_i$. The eigenvalue spectral gap of $C$ is then defined by $\Delta = 1 - \lambda$, where $\lambda$ is the largest non-unit eigenvalue of the product $G_1 G_2 G_3 G_2 G_1$. Define $\Delta'$ and $\lambda'$ analogously.
We first show the following.
\begin{lemma}
\label{lemma:exact_gaps}
In the limit of large $N$, we have $\lambda = \frac{1}{25}$ and $\lambda' = 0$. More generally, for finite $N$ and local dimension $q$,
\begin{align}
    \label{eq:exact_gaps}
    \lambda &= \left(\frac{q^{2N} - q^4}{q^{2N} - q^2}\frac{1}{q^{2} + 1}\right)^2 \\
    \lambda' &= \left(\frac{q^2 - 1}{q^{N} - q^{2-N}}\right)^{2}
\end{align}
\end{lemma}
\begin{proof}
    Let $S_{12}$ be the operator which swaps sites $1$ and $2$. By the symmetry of the Haar measure, we have $S_{12} G_2 = G_2 S_{12} G_2$. We can thus write $G_1 G_2 G_3 G_2 G_1 = G_1 (G_2 S_{12}) G_3 (S_{12} G_2) G_1$. But $G_3 = S_{12} G_1 S_{12}$, so we see $G_1 G_2 G_3 G_2 G_1 = G_1 G_2 G_1 G_2 G_1 = (G_1 G_2 G_1)^2$.

    It follows that the subleading eigenvalue of $G_1 G_2 G_3 G_2 G_1$ is the square of the subleading eigenvalue of $G_1 G_2 G_1$. This latter circuit is in effect a three-site brickwork with local dimensions $(q, q, q^{N-2})$.

    Now we use Theorem 37 of Ref.~\onlinecite{Allen2024}, which gives an eigenvalue spectrum of the form
    \begin{gather}
        \left\{ 1, \,\,\,\frac{(Q_1^2-1)Q_2^2(Q_3^2-1)}{(Q_1^2Q_2^2-1)(Q_2^2Q_3^2-1)}\right\}
    \end{gather}
    for a general three-site brickwork with local dimensions $(Q_1, Q_2, Q_3)$. Applying $Q_1=Q_2=q, Q_3=q^{N-2}$ then establishes the first line of Equation \ref{eq:exact_gaps}. 

    We next consider the circuit $G_1 G_3 G_1$. This is also effectively a three-site brickwork, but the local dimensions are now $(Q_1, Q_2, Q_3)=(q, q^{N-2}, q)$. We can thus follow the same calculation with different parameter values to obtain the second line of Equation \ref{eq:exact_gaps}. The remainder of the lemma is obtained by substituting $q = 2$ and taking the limit $N \rightarrow \infty$.
\end{proof}

To establish Theorem \ref{thm:depth}, we use the following result.
\begin{lemma}
    \label{lemma:depth_from_gap}
    Suppose \(\Phi\) is a moment operator of order $t$ acting on $N$ sites of dimension $q$ with largest non-unit eigenvalue $\lambda$. Then
    \begin{gather}
        \lambda \leq \epsilon_A \leq 2\epsilon_M \leq 2 q^{2Nt}\lambda
    \end{gather}
    where $\epsilon_A$ and $\epsilon_M$ are the additive and multiplicative errors, respectively, of \(\Phi\) relative to \(\Phi_\text{Haar}\).
\end{lemma}
\begin{proof}
    This statement is a slight generalization of Lemma 55 of ref.~\onlinecite{Allen2024}. We apply Lemma 4 of ref.~\onlinecite{Brandao2016}, which says that the multiplicative error $\epsilon_M$ between $\Phi$ and $\Phi_\text{Haar}$ is related to the spectral gap by $ \epsilon_M \leq q^{2Nt} \lambda$. We then apply Lemma 3 of the same work, which tells us \(\epsilon_A \leq 2 \epsilon_M\).

    By definition of the diamond norm,
    \begin{gather}
        \epsilon_A = ||\Phi - \Phi_\text{Haar}||_\diamond \geq \frac{||(\Phi - \Phi_\text{Haar})\rho||_1}{||\rho||_1}
    \end{gather}
    for all operators $\rho$. Consider the subleading eigenstate $\rho_*$ of the moment operator $\Phi$. We see
    \begin{gather}
        ||\Phi\rho_*||_1 = \lambda_* ||\rho_*||_1 = \lambda ||\rho_*||_1\\
    \end{gather}
    It also follows from Lemma 10 of ref.~\onlinecite{Belkin2023} that
    \begin{gather}
        ||\Phi_\text{Haar} \rho_*||_1 = 0
    \end{gather}
    and so \(\epsilon_A \geq \lambda\).
\end{proof}
Now, let $N$ be large enough that $\lambda > \lambda'$. We have
\[||\Phi^d - \Phi_{\text{Haar}}||_\diamond \geq \lambda^d\]
and 
\[||\Phi'^d - \Phi_{\text{Haar}}||_\diamond \leq q^{2Nt}{\lambda'}^d\]
When
$q^{2Nt}{\lambda'}^d < \lambda^d$
it follows that 
$||\Phi'^d - \Phi_{\text{Haar}}||_\diamond < ||\Phi^d - \Phi_{\text{Haar}}||_\diamond$. This occurs when the depth is
\begin{gather}
    d > \frac{2Nt \log q}{\log \frac{\lambda}{\lambda'}}
\end{gather}
The conditions of Theorem \ref{thm:depth} are thus satisfied by choosing
\begin{align}
    d &= 1 + \left \lfloor\frac{2Nt \log q}{\log \frac{\lambda}{\lambda'}}\right \rfloor \\
    \epsilon_A &= q^{2Nt}{\lambda'}^d
\end{align}
The same argument goes through for multiplicative error, with $\epsilon_M = \frac{1}{2}\epsilon_A$. For the architecture of Figure \ref{fig:architecture}a with $q = t = 2$, this formula becomes
\begin{align}
    d &= 1 + \left \lfloor
    \frac{2}{1 + \frac{1}{N}\log_2 \left|\frac{1 - 16 \cdot 4^{-N}}{15}\right|}\right \rfloor \\
    \epsilon_A &= 8^{N}\left(\frac{3}{2^N - 4 \cdot 2^{-N}}\right)^{2d}
\end{align}
We can also compute the exact multiplicative error numerically. Results are shown in Figure \ref{fig:architecture}c. We see that for $N \geq 5$, the censoring inequality is violated even at $d = 1$.

\section{Conclusion}
We have shown that censoring inequalities for quantum circuit architectures generally do not hold. In particular, gates which ``hide'' quantum information in a sheltered location can reduce scrambling. Our arguments apply to several measures of scrambledness and several types of deletion. Furthermore, the underlying intuition suggests that any other reasonable measure of scrambledness is likely to yield the same story. 

These results rule out certain approaches to resolving two open questions. The first question is the approximate $t$-design depth of the 1D brickwork, which is related to the architecture of Ref. \onlinecite{Schuster2024} by censoring. Second is the maximum possible singular value spectral gap of a connected architecture. Any connected architecture can be reduced to some spanning tree by deleting certain gates, so a censoring inequality would have implied that the maximum possible gap occurred in a treelike architecture.

Our results suggest a more general question about how global scrambling depends on local dynamics. In the approximate $t$-design setting, Ref.~\onlinecite{Yada2025} gives a relationship between local and global spectral gaps, suggesting that the rate of global scrambling is at worst proportional to the rate of local scrambling. In the same vein, Ref.~\onlinecite{Suzuki2024} shows that certain highly-entangling local ensembles give $t$-designs faster than Haar-random local gates. One can interpret this result as showing, loosely, that removing randomness from the circuit in a way that increases the typical amount of local scrambling makes the circuit a better approximate $t$-design. In this framing, our results show that even removing randomness in a way that \textit{decreases} local scrambling can have the same effect. It remains unclear if there is any clear conceptual description of the relationship between local and global behavior.

We also show that censoring inequalities do hold for the last gate in a circuit. The resulting crossover in scrambledness suggests that some unitary ensembles may be in some sense less resilient against scrambling than others, even at equal distance from the Haar measure. It is not clear whether this should be attributed only to the fine-tuning of these specific architectures or if there is some more general principle at work. What properties of an architecture determine the minimum number of additional gates necessary to obtain an $\epsilon$-approximate $t$-design? This may be a promising area for future study.

We give several examples of architecture which violate censoring inequalities. However, all of our architectures feature a connectivity graph with nontrivial topology. However, none of our architectures feature a linear connectivity graph. It thus remains plausible that censoring inequalities do hold for architectures with geometric locality on the line, since in this case it is more difficult for information to ``dodge'' scrambling.  Our examples also feature discrete unitary gates. It may be interesting to investigate a continuous-time version of this problem, such as the conditions under which adding a new interaction to a random Hamiltonian can slow down thermalization.

\bibliographystyle{unsrt}
\bibliography{references}

\begin{appendices}
\section{Numerical results}
\label{app:numerics}
\subsection{Exact multiplicative error}
\label{app:numerical_multerrs}
Figure \ref{fig:multerr-crossover} shows the results of a numerical calculation of the multiplicative error of the second moment operator of the architectures shown in Figure \ref{fig:architecture} relative to the Haar measure. 
\begin{figure}[h]
    \centering
    \begin{tikzpicture}
        \begin{scope}
            \node[anchor=north west,inner sep=0] (image_b) at (0,0)
            {\includegraphics[width=0.4\columnwidth]{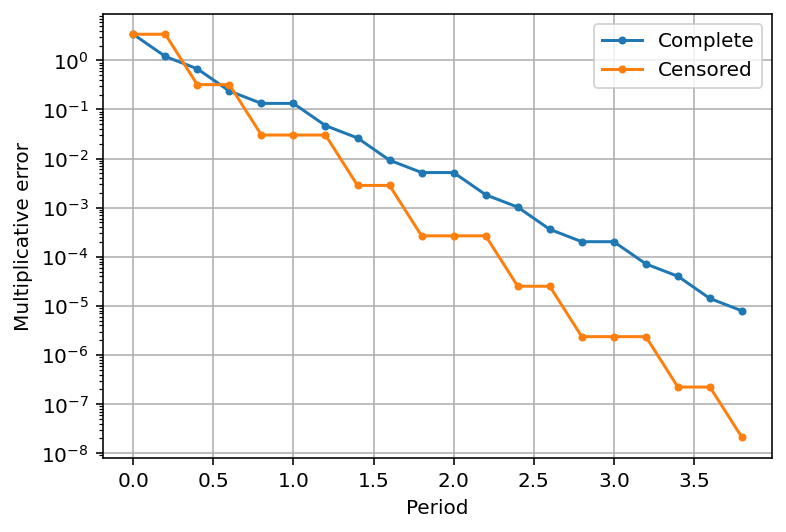}};
        \end{scope}
    \end{tikzpicture}
    \vspace*{-0.4cm}
    \caption{Crossover of multiplicative error for the circuit of Figure \ref{fig:architecture} with $5$ qubits.}
    \label{fig:multerr-crossover}
\end{figure}
One can see that omitting a gates causes the multiplicative error to increase in the short run, but decrease in the long run. This is similar in spirit to the Mpemba effect: A system which is further from equilibrium at early times approaches equilibrium faster at late times. A more dramatic instance of this effect is shown in Figure \ref{fig:mpemba}.

\begin{figure}[h]
    \centering
    \begin{tikzpicture}
        \begin{scope}
            \node[anchor=north west,inner sep=0] (image_a) at (0,0)
            {\includegraphics[width=0.3\columnwidth]{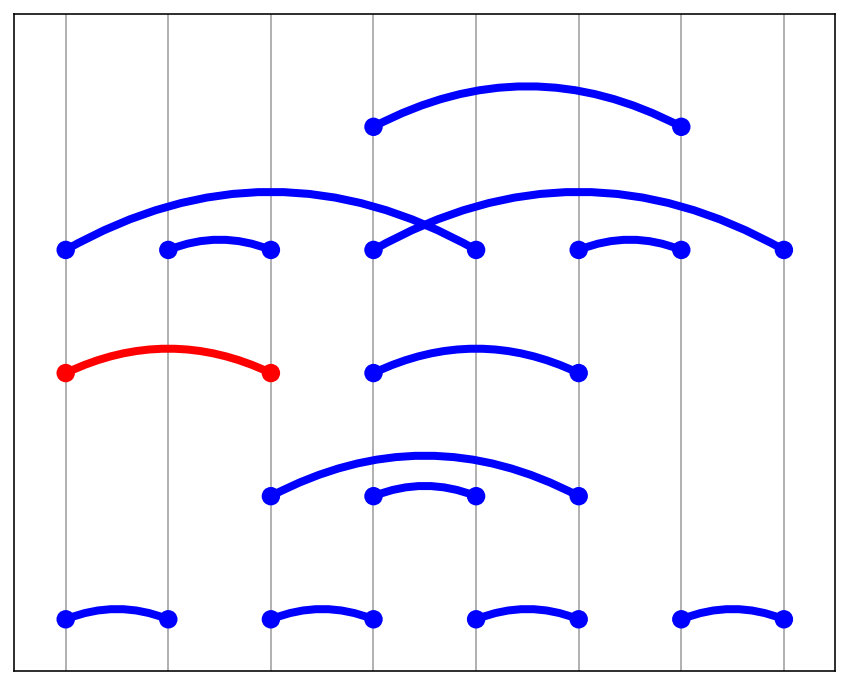}};
            \node [anchor=north west] (note) at (-0.5,0) {\small{\textbf{a)}}};
        \end{scope}
        \begin{scope}[xshift=0.33\columnwidth]
            \node[anchor=north west,inner sep=0] (image_b) at (0,0)
            {\includegraphics[width=0.4\columnwidth]{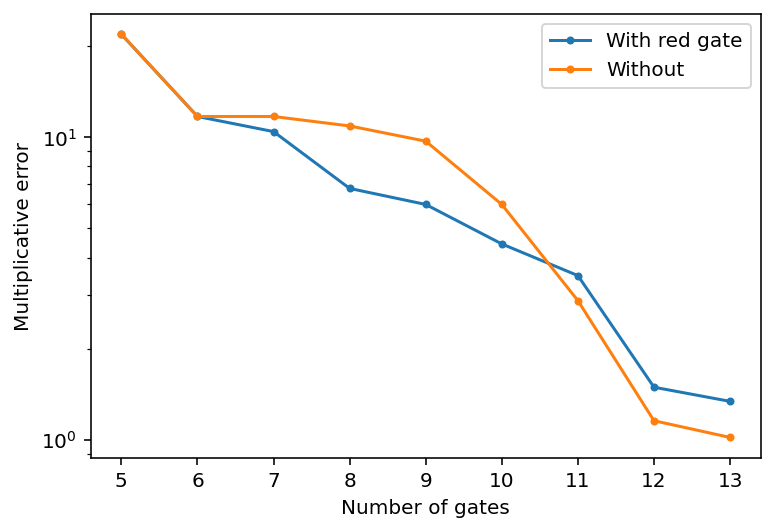}};
            \node [anchor=north west] (note) at (0.1,0) {\small{\textbf{b)}}};
        \end{scope}
    \end{tikzpicture}
    \vspace*{-0.4cm}
    \caption{(a) A circuit which shows a dramatic crossover in the multiplicative error over time. (b) Evolution of multiplicative error. Because the error is initially very large, we omit the first five points from this graph.}
    \label{fig:mpemba}
\end{figure}

\subsection{Small examples}
\label{app:minimal_examples}
Here we give smaller examples involving only two-site gates. They were found by numerical search; we have no analytical proof or intuition to explain their properties. The first, shown in Figure \ref{fig:minimal-circuit}, acts on five qubits with six gates and has a subleading eigenvalue which increases when the last gate is deleted. The second, shown in Figure \ref{fig:multerr-circuit}, acts on seven qubits with eight gates and has a multiplicative error which decreases when one of the interior gates is deleted. 

\begin{figure}[h]
    \centering
    \begin{tikzpicture}
        \begin{scope}
            \node[anchor=north west,inner sep=0] (image_a) at (0,0)
            {\includegraphics[width=0.4\columnwidth]{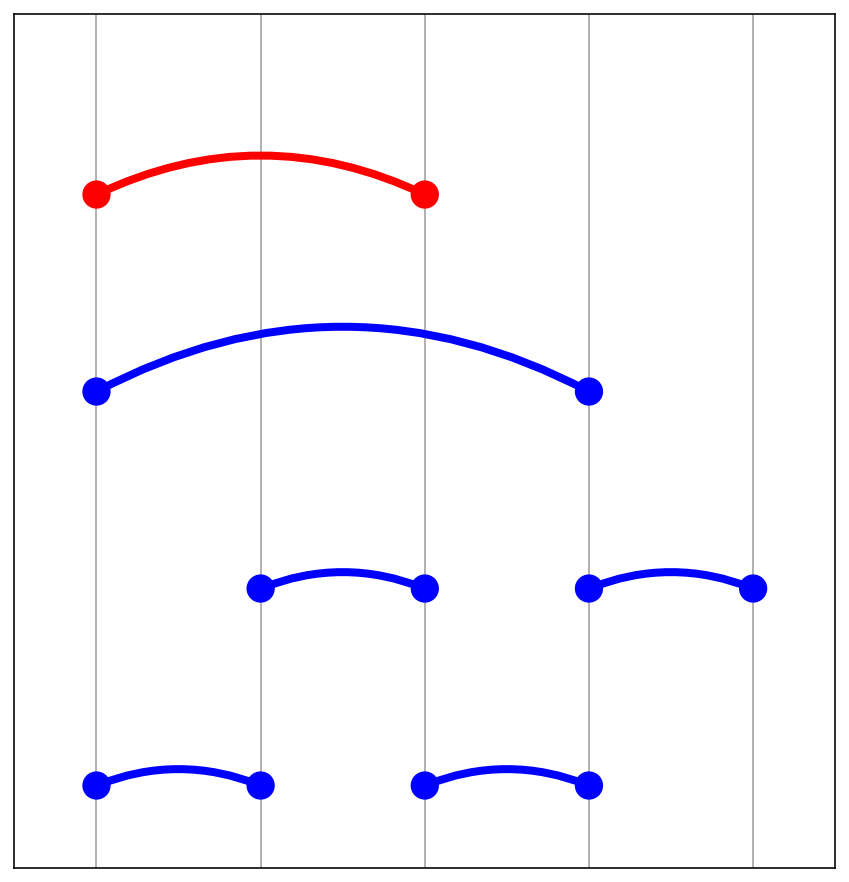}};
        \end{scope}
    \end{tikzpicture}
    \vspace*{-0.4cm}
    \caption{A small circuit on $5$ qubits which violates the censoring inequality for spectral gaps. The largest non-unit eigenvalue of the second moment operator is $0.2512$ with the final gate (red) and $0.2423$ without.}
    \label{fig:minimal-circuit}
\end{figure}

\begin{figure}[h]
    \centering
    \begin{tikzpicture}
        \begin{scope}
            \node[anchor=north west,inner sep=0] (image_b) at (0,0)
            {\includegraphics[width=0.4\columnwidth]{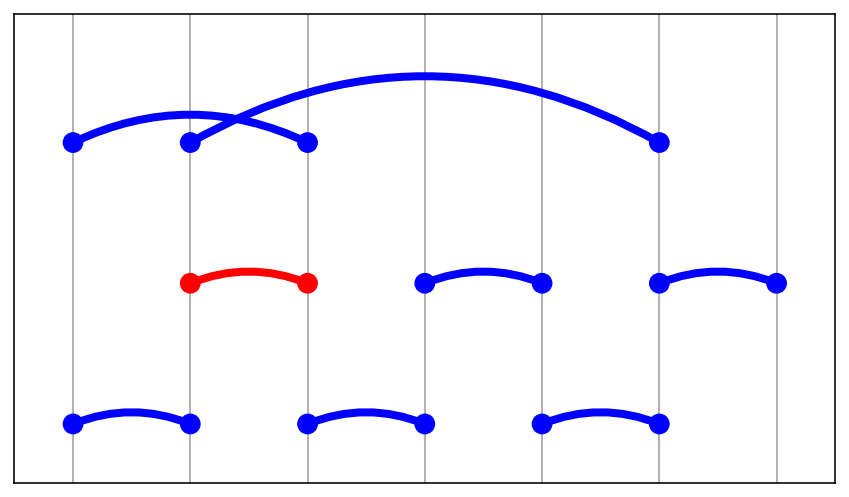}};
        \end{scope}
    \end{tikzpicture}
    \vspace*{-0.4cm}
    \caption{A small circuit on $7$ qubits which violates the censoring inequality for multiplicative error. The $t = 2$ multiplicative errors are $2.91$ and $1.96$ with and without the red gate, respectively.}
    \label{fig:multerr-circuit}
\end{figure}

\subsection{Architectures sampled from graphs}
\label{app:graphs}
Another common case of interest is architectures in which the spatial arrangement is not fixed. Instead, locations of gates are sampled from the uniform distribution over the edges of some graph over the sites. In this case one may ask if deleting an edge from the graph always decreases some measure of scrambledness. We show a counterexample to this conjecture by computing the spectral gaps of the second moment operator of the linear and lollipop graphs. The lollipop graph has more edges, but a smaller spectral gap. Results are shown in Figure \ref{fig:lollipop}.

\begin{figure}[h]
    \centering
    \begin{tikzpicture}
        \begin{scope}
            \node[anchor=north west,inner sep=0] (image_a) at (0,-0.5)
            {\includegraphics[width=0.35\columnwidth]{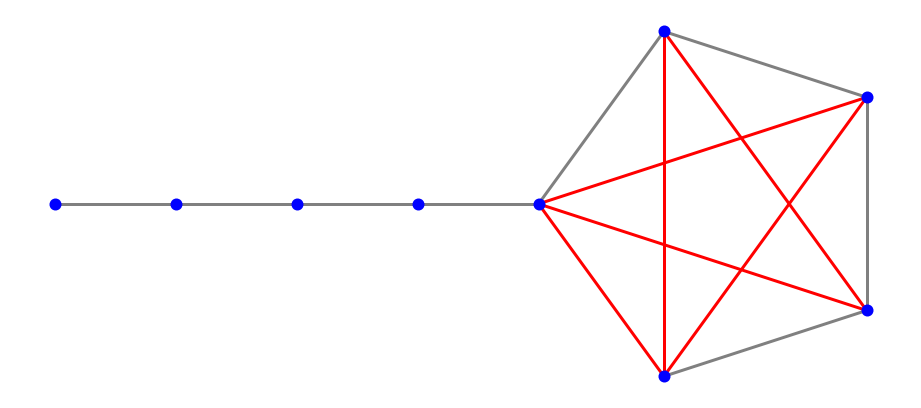}};
            \node [anchor=north west] (note) at (-0.0,0) {\small{\textbf{a)}}};
        \end{scope}
        \begin{scope}[xshift=0.42\columnwidth]
            \node[anchor=north west,inner sep=0] (image_b) at (0,0)
            {\includegraphics[width=0.35\columnwidth]{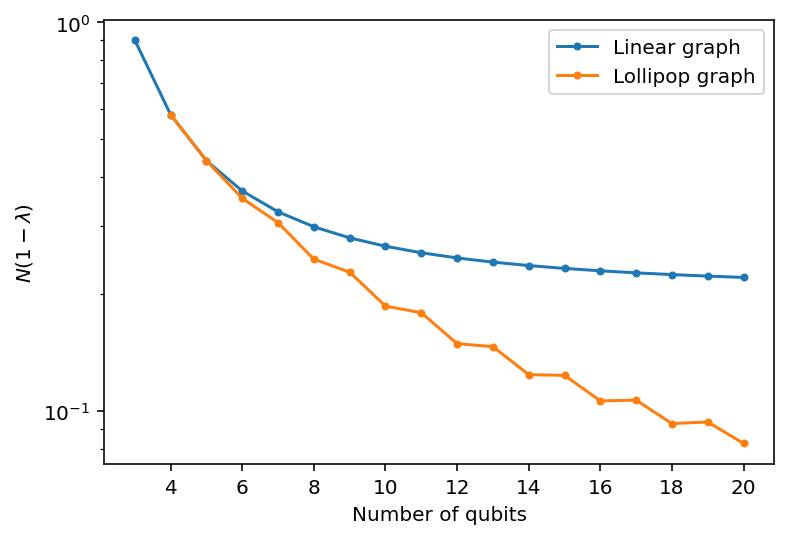}};
            \node [anchor=north west] (note) at (-0.15,0) {\small{\textbf{b)}}};
        \end{scope}
    \end{tikzpicture}
    \vspace*{-0.4cm}
    \caption{(a) The lollipop graph with 10 qubits. When the red edges are removed, the linear graph remains. (b) Normalized spectral gaps of the second moment operators of the linear and lollipop graph architectures.}
    \label{fig:lollipop}
\end{figure}
Since the moment operators corresponding to architectures sampled from graphs are Hermitian, the eigenvalue and singular value spectral gaps are the same. By Lemma \ref{lemma:depth_from_gap}, it follows that there exists some value of $\epsilon$ and number of gates such that the linear graph forms an $\epsilon$-approximate $t$-design, but the lollipop graph does not. This holds for both additive and multiplicative error.

\section{Proof of Table \ref{table:status}}
\label{app:table_proofs}
Here we describe the arguments which fill out Table \ref{table:status}. These proofs are all quite straightforward, but described explicitly here for completeness.

\begin{lemma}
Let \(\Psi\) be the moment operator corresponding to the last gate and \(\Phi'\) be that of the rest of the circuit, such that \(\Phi = \Psi \circ \Psi'\) is the moment operator corresponding to the full architecture. Let $s$ be the largest singular value of $\Phi - \Phi_H$, let $\epsilon_A = ||\Phi - \Phi_H||_\diamond$, and let $\epsilon_M$ be the smallest number such that \(\epsilon_M\Phi_H \pm (\Phi - \Phi_H)\) is completely positive. Define $s'$, $\epsilon_A'$, and $\epsilon_M'$ analogously. Then the following censoring inequalities hold:
\begin{align}
    s &\leq s' \\
    \epsilon_A &\leq \epsilon_A' \\
    \epsilon_M &\leq \epsilon_M' \\
\end{align}
\end{lemma}
\begin{proof}
\begin{enumerate}
    \item  Because $\Psi \circ \Phi_H = \Phi' \circ \Phi_H = \Phi_H \circ \Phi_H = \Phi_H$, we can write \(\Phi - \Phi_H = (\Psi - \Phi_H) \circ (\Phi' - \Phi_H)\). The largest singular values of $\Psi - \Phi_H$ and  $\Phi' - \Phi_H$ are $1$ and $s'$, respectively, so the largest singular value of their product is at most $s'$.
    \item The additive error is equal to the diamond norm distance from the Haar measure. But
    \begin{gather}
        ||\Phi - \Phi_H||_\diamond = ||\Psi \circ (\Phi' - \Phi_\text{Haar})||_\diamond 
    \end{gather}
    which by monotonicity(?) of the diamond norm is at most \(||\Phi' - \Phi_\text{Haar}||_\diamond\).
    \item By assumption \(\epsilon_M'\Phi_H \pm (\Phi' - \Phi_H)\) is completely positive. \(\Psi\) is also completely positive. The composition of two completely positive channels is also completely positive, and so 
\begin{gather}
    0 \preceq \Psi \circ \left[\epsilon_M'\Phi_H \pm (\Phi' - \Phi_H)\right] = \epsilon_M'\Phi_H \pm (\Phi - \Phi_H)
\end{gather}
is also completely positive. Since $\epsilon_M$ is defined to be the smallest number such that \(0 \preceq \epsilon_M\Phi_H \pm (\Phi - \Phi_H)\), we have $\epsilon_M \leq \epsilon_M'$.
\end{enumerate}
\end{proof}
These arguments apply to any gate from the first or last layer of an architecture, which we refer to as ``boundary gates".

For interior gates, Lemma \ref{lemma:exact_gaps} gives a counterexample to the eigenvalue spectral gap. Since the architecture of Figure \ref{fig:architecture} is Hermitian, this also serves as a counterexample for the singular value spectral gap. We can cyclically permute gates to make a deleted gate a boundary gate without changing the eigenspectrum, which is a counterexample for the eigenvalue spectral gap in the case of boundary gate deletion. Applying Lemma \ref{lemma:depth_from_gap} then establishes multiplicative and additive error.

For graph edges, meanwhile, the numerical calculations shown in Figure \ref{fig:lollipop} 
give a counterexample for the eigenvalue spectral gap. The moment operator is Hermitian, so this also serves for the singular value spectral gap. Applying Lemma \ref{lemma:depth_from_gap} again rules out a censoring inequality for multiplicative and additive error. This completes the proofs of Table \ref{table:status}.
\end{appendices}
\end{document}